%% file: main.tex
\RequirePackage{amsmath}
\documentclass{article}
\usepackage{graphicx}

\usepackage{tikz}
\usetikzlibrary{arrows,shapes,automata, calc, chains, matrix, trees, positioning, scopes}
\tikzset{snake it/.style={decorate, decoration=snake}}

\usepackage{amsthm,amsmath,amssymb,amsfonts}
\usepackage{bbm}
\usepackage{graphicx}
\usepackage{multirow}
\usepackage{xspace}
\usepackage[colorlinks=true,citecolor=blue,linkcolor=blue]{hyperref}
\usepackage{cleveref}

\newtheorem{theorem}{Theorem}
\usepackage{geometry}
\usepackage{fancyhdr}
\newcommand{\OO}{\mathcal{O}}
\newcommand{\Aa}{\mathcal{A}}

\newcommand{\PSPACE}{\textnormal{\textsf{PSPACE}}\xspace}

\newcommand{\mt}{\mathbf{mt}}

\newcommand{\mtd}{\mathbf{mt}_{\mathrm{d}}}

\newcommand{\comment}[1]{}

\DeclareMathOperator{\bin}{bin}

\title{On shortest products \\ for nonnegative matrix mortality}
\author{Andrew Ryzhikov}
\date{Department of Computer Science, University of Oxford, UK \\ {ryzhikov.andrew@gmail.com}}
\begin{document}

\maketitle              
\begin{abstract}
Given a finite set of matrices with integer entries, the matrix mortality problem asks if there exists a product of these matrices equal to the zero matrix. We consider a special case of this problem where all entries of the matrices are nonnegative. This case is equivalent to the NFA mortality problem, which, given an NFA, asks for a word $w$ such that the image of every state under $w$ is the empty set. The size of the alphabet of the NFA is then equal to the number of matrices in the set.
We study the length of shortest such words depending on the size of the alphabet. We show that for an NFA with $n$ states this length can be at least $2^n - 1$ for an alphabet of size $n$, $2^{(n - 4)/2}$ for an alphabet of size~$3$ and $2^{(n - 2)/3}$ for an alphabet of size $2$. We also discuss further open problems related to mortality of NFAs and DFAs.

\end{abstract}
\section{Introduction}
\input{sec-intro}\label{sec-intro}

\section{Existing results and their adaptations}\label{sec-existing}
\input{sec-existing}

\section{NFA mortality lower bounds}\label{sec-lb}
\input{sec-nfas}

\section{Why do automata die slowly?}\label{sec-dfas}
\input{sec-dfas}

\section{Conclusions and open problems}\label{sec-conclusions}
\input{sec-conclusions}




\subsubsection*{Acknowledgments} The author greatly benefited from talking to (and sometimes at) Andrei Draghici. The comments of anonymous reviewers improved the presentation of the paper. The SynchroViewer software (github.com/marekesz/synchroviewer) was helpful for testing conjectures about shortest mortal words for DFAs. The author is supported by the European Research Council (ERC) under the European Union’s Horizon 2020 research and innovation programme (Grant agreement No. 852769, ARiAT).

%
%
%
 \bibliographystyle{alpha}
 \bibliography{ref}
\end{document}

%% file: sec-intro.tex
Given a finite set $\Aa$ of integer square matrices of the same dimension, the matrix mortality problem asks if the monoid generated by $\Aa$, that is, the set of all possible products of matrices from $\Aa$, contains the zero matrix. This problem is undecidable already for sets of $3 \times 3$ matrices \cite{Paterson1970}, and is a subject of extensive research, see \cite{Cassaigne2014} for a survey of recent developments.

We consider a special case of matrix mortality where all matrices in $\Aa$ are nonnegative, that is, all their entries are nonnegative. We call this problem nonnegative matrix mortality. The theory of nonnegative matrices has applications in the theory of codes \cite{Berstel2010}, Markov chains \cite{Seneta2006} and symbolic dynamics~\cite{Lind2021}, and recently there has been a significant interest in reachability problems for sets of nonnegative matrices and monoids generated by them \cite{Protasov2012,Gerencser2018,Wu2023}.

Clearly, for nonnegative matrix mortality it only matters which entries of the matrices in $\Aa$ are strictly positive and which are zero. In particular, this implies that nonnegative matrix mortality is decidable in polynomial space. Moreover, for a set $\Aa$ of nonnegative $n \times n$ matrices, the length of a shortest product from this set resulting in the zero matrix is at most $2^n - 1$. This easily follows from the interpretation of $\Aa$ as a non-deterministic finite (semi-)automaton (NFA), as explained, e.g., in~\cite{Rampersad2012}. Since our results, as well as most of the existing literature on nonnegative matrix mortality, use this correspondence, we provide its details. 

We view a finite set $\Aa$ of nonnegative $n \times n$ matrices as an NFA with a state set $Q = \{1, 2, \ldots, n\}$. Each letter $a_j$ of its alphabet $\Sigma$ corresponds to a matrix $A_j \in \Aa$. To define the transition relation $\Delta: Q \times \Sigma \to 2^Q$, we take $\Delta(i, a_j)$ to be the set of all $h$ such that the entry $(i, h)$ in $A_j$ is strictly positive. When the transition relation is clear from the context, we denote $\Delta(q, a)$ as $q \cdot a$. We homomorphically extend it to words over $\Sigma$, and then to subsets $S \subseteq Q$ as $S \cdot w = \{p \mid p \in q \cdot w \mbox{ for some } q \in S\}$. In thus constructed NFA $(Q, \Sigma, \Delta)$, with respect to matrix mortality, 
words correspond to products of matrices from $\Aa$: we have $Q \cdot w = \emptyset$ if and only if the product of matrices corresponding to $w$ is the zero matrix. Such words $w$ are called \emph{mortal}. To obtain the mentioned upper bound of $2^n - 1$ (observed, e.g., in \cite{Imreh1999}), it is now enough to note that for a shortest mortal word $w = x_1 x_2 \ldots x_k$ the sets $Q \cdot x_1 x_2 \ldots x_i$ for different values of $i$ are pairwise different subsets of~$Q$, and hence $k \le 2^n - 1$.

The length of a shortest mortal word of an NFA is called its \emph{mortality threshold}. Let $\mt(n, m)$ denote the maximum mortality threshold of an NFA with~$n$ states and $m$ letters admitting a mortal word. In this paper, we investigate the values $\mt(n, m)$ depending on $m$. As shown above, for every $n$ and $m$, $\mt(n, m) \le 2^n - 1$. We survey the known bounds on $\mt(n, m)$ in the next section, and in the rest of this section we mention some other relevant~results.

The problem of checking if a given NFA admits a mortal word is \PSPACE-complete, even when restricted to binary NFAs \cite{Kao2009}. As usually, we call an NFA \emph{binary} if its alphabet has size two, and \emph{ternary} if it has size three.

In \cite{Kiefer2021}, the case where $\Aa$ is a set of nonnegative integer matrices of joint spectral radius at most one is considered. This case includes, in particular, all sets~$\Aa$ of $\{0, 1\}$-matrices (that is, matrices whose entries belong to the set $\{0, 1\}$) with the property that the monoid generated by~$\Aa$ contains only $\{0, 1\}$-matrices. Such sets $\Aa$ correspond to unambiguous NFAs, which are NFAs with the property that for every pair $p, q$ of states and every word $w$ there is at most one path from $p$ to~$q$ labelled by $w$.
It is proved in \cite{Kiefer2021} that for sets of joint spectral radius at most one, the existence of a product resulting in the zero matrix can be checked in polynomial time, and for all~$m$ we have $\mt_{\rho \le 1}(n, m) \le n^5$, where $\mt_{\rho \le 1}(n, m)$ is $\mt(n, m)$ restricted to such sets. The best known lower bound is quadratic in~$n$~\cite{Rystsov1997}. Some more results on mortal words for unambiguous NFAs are presented in \cite{Boccuto2019}.

\medskip

\textbf{Our contributions.}
We study the following question: given a set of $m$ nonnegative $n \times n$ matrices, how long can their shortest product resulting in the zero matrix be? In \Cref{sec-existing} we survey known results. \Cref{sec-lb} contains our main contributions. Namely, for $m = n$, we show a lower bound of $2^n - 1$ (\Cref{thm-linear-alph}), matching the upper bound for arbitrary $m$. We prove it by constructing an NFA simulating a binary counter counting from $2^n - 1$ to zero. The case of constant $m$ uses the same idea, but requires much more involved constructions to implement it, resulting in a lower bound of $2^{(n - 4)/2}$ for $m = 3$ (\Cref{thm-ternary}) and a lower bound of $2^{(n - 2)/3}$ for  $m = 2$ (\Cref{thm-binary}).
We conclude the paper by 
discussing other approaches to NFA and DFA mortality in \Cref{sec-dfas}, and providing a list of open problems related to this setting in \Cref{sec-conclusions}.

%% file: sec-existing.tex
In this section, we explain existing results on $\mt(n, m)$. The first result about this value seems to be in the work \cite{Goralcik1982}, where it is proved that $\mt(n, m)$ is not bounded by any polynomial in $n$. Most of the known results are not stated in terms of matrix or NFA mortality, so we appropriately rephrase them. We also show that strong lower bounds on $\mt(n, m)$ can be easily obtained from existing results that are not obviously related to mortality. 

\subsection{Factorial languages and shortest rejected words}

Clearly, mortal words for an NFA are precisely the words that do not label a path between any two states in the NFA. Equivalently, these are words that are not accepted by the NFA if we make each state both initial and final. Shortest such words were studied in 
\cite{Kao2009,Rampersad2012}, building upon an earlier work  \cite{Ellul2004} about a more general question: what is the maximum length of a shortest word not accepted by an NFA with a given number of states? The result relevant to our work is as follows.

\begin{theorem}[Theorem 16 in \cite{Kao2009}, rephrased] For every $n$, there exists an NFA with $n$ states over a ternary alphabet such that its shortest mortal word is of length at least $2^{n/75}p(n)$, where $p(n)$ is some polynomial.
\end{theorem}

\subsection{D1-directing and carefully synchronizing words}

A word $w$ is called \emph{D1-directing} for an NFA if there exists a state $q$ such that for every state $p$ we have $p \cdot w = \{q\}$ \cite{Imreh1999}. An NFA is called \emph{total}\footnote{Such NFAs are often called complete, for example in \cite{Imreh1999,Imreh2003}. However, in e.g. \cite{Restivo1981,Mika2021} the term ``complete'' refers to objects corresponding to matrix monoids that do not contain the zero matrix, which we find reasonable. We thus find the term ``total'' more appropriate, since an NFA is total if and only if the action of every its letter is a total binary relation. A total NFA is then necessarily complete (that is, does not admit a mortal word), but the opposite does not always hold, as illustrated by the following NFA with two letters:
\raisebox{-0.25\height}{\scalebox{.7}{\begin{tikzpicture} [node distance = 2cm]
\tikzset{every state/.style={inner sep=1pt,minimum size=1.5em}}
\node [state] at (0, 0) (1) {$1$};
\node [state] at (2, 0) (2) {$2$};
\path [-stealth, thick]
(1) edge [bend left=20] node[above] {} (2)
(2) edge [dashed, bend left=20] node[above] {} (1)
(1) edge [loop left] node[above] {} (1)
(2) edge [dashed, loop right] node[above] {} (2)
;
\end{tikzpicture}}}.
} if for every state $q$ and every letter~$a$ we have $q \cdot a \ne \emptyset$. A state $q$ in an NFA is called a \emph{sink} state (called a trap state in \cite{Imreh2003}) if for every letter~$a$ we have $q \cdot a = \{q\}$. A word is D1-directing for a total NFA $\Aa$ with a sink state $q$ if and only if it is mortal for the NFA obtained from $\Aa$ by removing $q$. We thus get the following.

\begin{theorem}[Theorem 3 in \cite{Imreh2003}]
    For every $n \ge 3$, $\mt(n, 2^n - 5) = 2^n - 1.$
\end{theorem}

In \Cref{thm-linear-alph} below we show that only $n$ letters are enough to achieve the same lower bound, thus obtaining that for every $n \ge 1$, $\mt(n, n) = 2^n - 1$.


The most well-studied setting of D1-directing words for NFAs is that of DFAs, in which case such words are often called \emph{carefully synchronizing}. Recall that a DFA is an NFA such that $|q \cdot a| \le 1$ for every state $q$ and every letter $a$. As~usually, we use a partial function $\delta: Q \times \Sigma \rightharpoonup Q$ instead of $\Delta: Q \times \Sigma \to 2^Q$ to highlight this. A DFA admitting a carefully synchronizing word is also called \emph{carefully synchronizing}. Strong lower bounds are known for the length of shortest carefully synchronizing words for DFAs, and these bounds can be transferred to  shortest mortal words for NFAs as follows. 

Let $\Aa = (Q, \Sigma, \delta)$ be a carefully synchronizing DFA. Consider an NFA $\Aa' = (Q, \Sigma \cup \{r\}, \Delta')$ obtained from $\Aa$ as follows: for every pair $q \in Q, a \in \Sigma$ such that $\delta(q, a)$ is undefined, take $\Delta'(q, a) = Q$. If $\delta(q, a)$ is defined, take $\Delta'(q, a) = \{\delta(q, a)\}$. Pick a state $p \in Q$ such that there exists a carefully synchronizing word~$w$ mapping all states to $p$. Take $\Delta(p, r) = \emptyset$ and $\Delta(q, r) = Q$ for all~$q \ne p$.  If $w$ is a carefully synchronizing word for $\Aa$, then $wr$ is a mortal word for $\Aa'$. Conversely, by construction, every mortal word for $\Aa'$ must contain a carefully synchronizing word for $\Aa$ as a factor. 
By applying this transformation to Theorems 14 and 16 from 
\cite{Bondt2019}, we obtain the following.

\begin{theorem}[\cite{Bondt2019}]\label{thm-careful}
    We have $\mt(n, 4) = \Omega(\frac{2^{2n/5}}{n})$ and $\mt(n, 3) = \Omega(\frac{2^{n/3}}{n\sqrt{n}})$.
\end{theorem}

Note that sometimes adding a new letter $r$ is not necessary in the construction described above, which seems to be the case for \cite{Bondt2019}. However, proving that requires analysing the whole construction for careful synchronization, which is quite involved. Since our results on NFA mortality are stronger than those in \Cref{thm-careful} even without the additional letter, we do not pursue this any further. We also remark that the applicability of this approach is limited by the upper bound on shortest carefully synchronizing words for $n$-state DFAs, which is $3^{(1 + \epsilon)n/3}$ for every $\epsilon > 0$ \cite{Rystsov1980}. This upper bound is asymptotically tight already for an alphabet of linear size \cite{Rystsov1980,Martyugin2010,Ryzhikov2018}.

%% file: sec-nfas.tex
For all lower bounds in this section, the main idea is to construct an NFA simulating a binary counter. Namely, for every $k$ we construct an NFA with the number of states linear in $k$ such that, when reading a word $w$, a dedicated set $Q = \{q_1, \ldots, q_k\}$ of $k$ states behaves in the following way. We say that a state $q$ is \emph{active} after the application of a word $w$ if $q \in p \cdot w$ for some state $p$, otherwise we say that it is \emph{inactive}. Initially, all states in $Q$ are active. We interpret active and inactive states from $Q$ as a number in the binary encoding.
Namely, given a word $w \in \Sigma^*$, define $\bin(w) = x_1 x_2 \ldots x_k$, where $x_i = 1$ if $q_i \in Q \cdot w$, and $x_i = 0$ otherwise. We treat $\bin(w)$ as a natural number represented in the binary notation, with the most significant bit first.
By construction of the NFA we guarantee that when reading a letter, this number is either decremented by at most one or is set again to the number consisting of all ones. The length of a shortest mortal word for such NFA is then at least $2^k - 1$. 

For an alphabet whose size is equal to the number of states, this idea can be implemented in a fairly straightforward way with a $k$-state NFA. To implement it with a ternary and a binary alphabet, 
we need more advanced techniques. The main challenge is that at least one of the letters must have a different effect on the encoding of the number depending on its value. To achieve that, between every two decrements we shift the number to the right until it ends with a one, thus iteratively dividing it by two, and simultaneously add some additional marker bits to preserve the information about the number of shifts we performed. In particular, it is important to remark that in these constructions not every application of a letter decrements the value of the number that we are tracking. Namely, we will have one letter that actually decrements the value, and all other letters will perform auxiliary operations such as shifting the representation. These ideas will be explained in detail in due course.


\textbf{Conventions for figures.} In all figures in this section, we use a vertical bold outgoing arrow to indicate that there is a transition from the state it originates from to all the states of the NFA. Similarly, a vertical bold incoming arrow indicates that there is a transition from each state of the NFA (except for states with a dashed outgoing arrow) to the state where the arrow ends. A dashed outgoing arrow from a state means that the image of this state is empty.



\subsection{Linear-size alphabet is enough for a tight lower bound}

We start by proving the following statement. As its consequence, together with the upper bound mentioned in the introduction, we get that $\mt(n, n) = 2^n - 1$.

\begin{theorem}\label{thm-linear-alph}
    For every integer $n \ge 1$, there exists an NFA with $n$ states and $n$ letters such that the length of its shortest mortal word is $2^n - 1$.
\end{theorem}

\begin{proof}
    
For every positive $n$ we construct an NFA $\Aa = (Q, \Sigma, \Delta)$ as follows. We take $Q = \{q_1, \ldots, q_n\}$ (thus we have $k = n$ in the informal explanations in the beginning of this section). Let $\Sigma = \{a_1, \ldots, a_n\}$. Define 
\[
    q_i \cdot a_j = 
        \begin{cases}
         \emptyset & \mbox{if } i = j = n; \\
        \{q_{i + 1}, \ldots, q_n\} & \mbox{if } 1 \le i = j < n; \\
        Q  & \mbox{if } 1 \le j < i \le n; \\
        \{q_i\} & \mbox{if } 1 \le i < j \le n. \end{cases}
\]

For $n = 5$, the action of $a_2$ is illustrated below:

\begin{center}
\begin{tikzpicture} [node distance = 2cm]
\tikzset{every state/.style={inner sep=1pt,minimum size=1.5em}}

\node [state] at (0, 0) (1) {$q_1$};
\node [state] at (2, 0) (2) {$q_2$};
\node [state] at (4, 0) (3) {$q_3$};
\node [state] at (6, 0) (4) {$q_4$};
\node [state] at (8, 0) (5) {$q_5$};

\node [] at (0, 1) (c1) {};
\node [] at (2, 1) (c2) {};
\node [] at (4, 1) (c3) {};
\node [] at (6, 1) (c4) {};
\node [] at (8, 1) (c5) {};

\path [-stealth, thick]

(3) edge [line width=2pt] node[above] {} (c3)
(4) edge [line width=2pt] node[above] {} (c4)
(5) edge [line width=2pt] node[above] {} (c5)

(2) edge [bend right=20] node[above] {} (3)
(2) edge [bend right=20] node[above] {} (4)
(2) edge [bend right=20] node[above] {} (5)

(1) edge [loop below] node[above] {} (1)

;

\end{tikzpicture}
\end{center}


First we note that there exists a mortal word for $\Aa$. We construct it letter by letter. At each step we apply letter $a_j$, where $j$ is such that state $q_j$ is currently active, and states $q_{j + 1}, \ldots, q_n$ are not active. In other words, $j$ is the position of the last 1 in the representation of active states of $Q$ as a number in binary. The length of thus constructed word is $2^n - 1$. This is in fact the only mortal word of this length, and to better illustrate its structure, let us describe how it can be constructed recurrently. Take $w_n = a_n$, and for all $1 \le i \le n - 1$ take $w_i = w_{i + 1} a_i w_{i + 1}$. The word $w_1$ is then the shortest mortal word for $\Aa$. The shortest mortal word for the example above for $n = 5$ is
\[a_5 a_4 (a_5) a_3 (a_5 a_4 a_5) a_2 (a_5 a_4 a_5 a_3 a_5 a_4 a_5) a_1 (a_5 a_4 a_5 a_3 a_5 a_4 a_5 a_2 a_5 a_4 a_5 a_3 a_5 a_4 a_5),\] 
where the partition of the word indicates the second appearance of $w_{i + 1}$ in each step of the recurrence $w_i = w_{i + 1} a_i w_{i + 1}$.

Now we show that this is the shortest mortal word for $\Aa$. Let $w$ be a mortal word for $\Aa$. As mentioned above, we consider $\bin(w')$ for prefixes $w'$ of $w$. We start with $\bin(\epsilon) = 2^n - 1$. By construction, reading a letter can decrement the value of the number by at most one: for every $a \in \Sigma$, we have $\bin(w'a) \ge \bin(w') - 1$.
Indeed, let $a = a_j$ and let $\bin(w') = x_1 \ldots x_n$ with $x_1, \ldots, x_n \in \{0, 1\}$. If $x_{j + 1} = \ldots = x_n$, then $\bin(w'a) = x_1 \ldots x_{j - 1} 0 1 1 \ldots 1$. If $x_j = 1$, we have $\bin(w'a) = \bin(w') - 1$, otherwise $\bin(w'a) > \bin(w')$. Now consider the case where $x_i = 1$ for some $i > j$. Then $\bin(w'a) = 2^n - 1$ by construction. Hence, the length of a shortest mortal word for $\Aa$ is at least $2^n - 1$.
\end{proof}

\subsection{Ternary case lower bound}

Now we implement a similar idea in a much more restricted setting of only three letters. We remark that combining the result from the previous subsection with a standard technique for reducing the size of the alphabet (for example, via a composition with a finite prefix code) is not enough to obtain a good lower bound. Indeed, going from a linear-size alphabet to an alphabet of size three requires appending to each state a tree with a linear number of inner states to simulate a linear number of choices (represented by the leaves of this tree). This results in a quadratic increase of the number of states. A possible improvement of this approach might be then to reuse the trees between different states in a clever way. However, we were not able to make it work, and our approach is instead to use shifts that directly change the binary representation of the number that we are tracking. A significant part of the construction is thus devoted to guaranteeing that this number cannot be decremented too much by applying a short word.
The goal of this subsection is to prove the following.

\begin{theorem}\label{thm-ternary}
    For every even $n \ge 6$, there exists an NFA with $n$ states and $3$ letters such that the length of its shortest mortal word is at least $2^{(n - 4)/2}$.
\end{theorem}

\begin{proof}
We construct an NFA $\Aa = (P \cup Q \cup \{f\}, \{s, d, c\}, \Delta)$.
Its set of states consists of a set of $k + 1$ states $P = \{p_0, p_1, \ldots, p_k\}$, which we refer to as the left half, a set of $k$ states $Q = \{q_1, \ldots, q_k\}$, which we refer to as the right half, and a special ``flag'' state $f$. We thus have $n = 2k + 2$.

As mentioned above, we will be tracking the values $\bin(w')$ for prefixes $w'$ of an input word~$w$. However, since we only have three letters, we will have to use a more involved approach to decrementing these values by one. Namely, when applying letters to the NFA, the encoding of the tracked number will sometimes be shifted. In more detail, this encoding will be represented by a continuous segment of $k$ states in the sequence $p_1, p_2, \ldots, p_k, q_1, q_2, \ldots, q_k$. Most other states in this sequence will be inactive, except for one state remembering by how much the representation was shifted. 
The main reason why we need to shift the number, and therefore why we need more states, is that, just like the action of $a_j$ is different for different positions in the proof of \Cref{thm-linear-alph}, 
the action of one of the three available letters must have a different effect depending on the current position of the last $1$ in the encoding of the number.

We first define the action of letter $s$ that performs shifting. Its initial purpose is to make active precisely the set $Q \cup \{f\}$. After that, it will allow shifting $\bin(w)$ to the right until its last digit is~1. If another shift is performed after that, the set $Q \cup \{f\}$ gets reactivated, so the whole process is reset to the beginning. The actions of letters $d$ and $c$ defined later will guarantee that the number of removed zeroes at the end of $\bin(w)$ is remembered elsewhere, and hence the value of the tracked number is not lost.

Formally, we define
\[
    p_i \cdot s = 
        \begin{cases}
         \{p_{i + 1}\} & \mbox{if } 0 \le i < k; \\
        \{q_1\} & \mbox{if } i = k; \end{cases}\qquad
    q_i \cdot s = 
        \begin{cases}
         \{q_{i + 1}\} & \mbox{if } 1 \le i < k; \\
        Q \cup \{f\} & \mbox{if } i = k; \end{cases}
\]
\[
f \cdot s = \{f\}.
\]

The action of letter $s$ for $k = 4$ is illustrated below:

\begin{center}
\begin{tikzpicture} [node distance = 2cm]
\tikzset{every state/.style={inner sep=1pt,minimum size=1.5em}}

\node[draw] at (0.7,-1) {action of $s$};

\node [state] at (0, 0) (0) {$p_0$};
\node [state] at (1, 0) (1) {$p_1$};
\node [state] at (2, 0) (2) {$p_2$};
\node [state] at (3, 0) (3) {$p_3$};
\node [state] at (4, 0) (4) {$p_4$};

\node [state] at (6, 0) (q1) {$q_1$};
\node [state] at (7, 0) (q2) {$q_2$};
\node [state] at (8, 0) (q3) {$q_3$};
\node [state] at (9, 0) (q4) {$q_4$};

\node [state] at (5, -1) (f) {$f$};

\path [-stealth, thick]

(0) edge [] node[above] {} (1)
(1) edge [] node[above] {} (2)
(2) edge [] node[above] {} (3)
(3) edge [] node[above] {} (4)
(4) edge [] node[above] {} (q1)
(q1) edge [] node[above] {} (q2)
(q2) edge [] node[above] {} (q3)
(q3) edge [] node[above] {} (q4)

(q4) edge [loop above] node[above] {} (q4)
(q4) edge [bend right=30] node[above] {} (q3)
(q4) edge [bend right=30] node[above] {} (q2)
(q4) edge [bend right=30] node[above] {} (q1)

(q4) edge [bend left=20] node[above] {} (f)

(f) edge [loop below] node[above] {} (f)

;

\end{tikzpicture}
\end{center}

Next we define the action of letter $c$ that checks that the representation of the number that we are tracking was shifted to the right half and removes the ``marker'' bit at the end of this representation. This ``marker'' bit is a detail that we have not mentioned before, but which is crucial for the construction. Most of the time, the encoding of the number we are tracking will end with an additional~$1$ that prevents $s$ from changing the number without resetting the counting. However, once $c$ is applied, which can be done only when the number is shifted to the right half, this bit is removed, and we can remove zeroes at the end of the tracked number by applying $s$. To stop the tracked number from decreasing fast, an application of $c$ activates state $p_0$. After that, by construction, letter $c$ cannot be applied as long as there are active states in the left half, which prevents us from using it again to remove more ones at the end of the tracked number.

The number of shifts performed by $s$ after an application of $c$ is then equal to the index of the unique active state among the states $p_0, p_1, \ldots, p_k$. The action of $d$ (defined later) will take that into account when performing the decrement and returning the number to its proper representation.

Letter $c$ also deactivates state $f$. This state is used to disallow using $d$ too often: as we define later, if $f$ is active, using $d$ will reactivate all states. 

Formally, we define
\[
    p_i \cdot c = 
         P \cup Q \cup \{f\}  \mbox{ for } 0 \le i \le k; \qquad
    q_i \cdot c = 
        \begin{cases}
         \{q_i, p_0\} & \mbox{if } 1 \le i < k; \\
        \emptyset & \mbox{if } i = k; \end{cases}
\]
\[
    f \cdot c = 
         \emptyset.
\]

The action of letter $c$ is illustrated below:

\begin{center}
\begin{tikzpicture} [node distance = 2cm]
\tikzset{every state/.style={inner sep=1pt,minimum size=1.5em}}

\node[draw] at (0.7,-1.5) {action of $c$};

\node [] at (0, 1) (c0) {};
\node [] at (0, -1) (cc0) {};

\node [] at (1, 1) (c1) {};
\node [] at (2, 1) (c2) {};
\node [] at (3, 1) (c3) {};
\node [] at (4, 1) (c4) {};

\node [state] at (0, 0) (0) {$p_0$};
\node [state] at (1, 0) (1) {$p_1$};
\node [state] at (2, 0) (2) {$p_2$};
\node [state] at (3, 0) (3) {$p_3$};
\node [state] at (4, 0) (4) {$p_4$};

\node [state] at (6, 0) (q1) {$q_1$};
\node [state] at (7, 0) (q2) {$q_2$};
\node [state] at (8, 0) (q3) {$q_3$};
\node [state] at (9, 0) (q4) {$q_4$};

\node [] at (10, -1) (k4) {};

\node [state] at (5, -1) (f) {$f$};

\node [] at (6, -2) (kf) {};

\path [-stealth, thick]

(cc0) edge [line width=2pt] node[above] {} (0)

(0) edge [line width=2pt] node[above] {} (c0)
(1) edge [line width=2pt] node[above] {} (c1)
(2) edge [line width=2pt] node[above] {} (c2)
(3) edge [line width=2pt] node[above] {} (c3)
(4) edge [line width=2pt] node[above] {} (c4)

(q1) edge [loop above] node[above] {} (q1)
(q2) edge [loop above] node[above] {} (q2)
(q3) edge [loop above] node[above] {} (q3)

(q4) edge [dashed, bend left=40] (k4)
(f) edge [dashed, bend left=40] (kf)

;

\end{tikzpicture}
\end{center}

Finally, we define the action of letter $d$ that decrements the tracked value by one and simultaneously shifts its representation by $k$ positions to the left. We make sure that between any two applications of $d$ there must be at least one application of $c$. It is guaranteed by the fact that $d$ activates $f$, and if $f$ is already active, an application of $d$ makes all the states active.
Hence, before $d$ is applied, the representation of the tracked number is  deconstructed into two parts: the number of performed shifts is remembered in the unique active state of the left half, and the rest of the number is in the right half. An application of $d$ then subtracts one from the value of the number, and composes the representation of the number back to the proper binary representation.
Besides that, the action of $d$ adds back the ``marker'' bit at the end of the representation, which was removed by an application of $c$. 

Formally, we define 
\[
    p_i \cdot d = 
        \begin{cases}
        P \cup Q \cup \{f\} & \mbox{if } i = 0;\\
         \{q_1, \ldots, q_i\} \cup \{f\} & \mbox{if } 1 \le i < k; \\
        \{f\} & \mbox{if } i = k; \end{cases}\quad
    q_i \cdot d = 
        \begin{cases}
         \{p_i\} \cup \{f\} & \mbox{if } 1 \le i < k; \\
        \emptyset & \mbox{if } i = k; \end{cases}
\]
\[
f \cdot s = P \cup Q \cup \{f\}.
\]

The action of $d$ for some states is illustrated below. The transitions from $p_1$ and $p_2$ are omitted for the clarity of presentation.

\begin{center}
\begin{tikzpicture} [node distance = 2cm]
\tikzset{every state/.style={inner sep=1pt,minimum size=1.5em}}

\node[draw] at (0.7,-1) {action of $d$};

\node [] at (0, 1) (c0) {};

\node [state] at (0, 0) (0) {$p_0$};
\node [state] at (1, 0) (1) {$p_1$};
\node [state] at (2, 0) (2) {$p_2$};
\node [state] at (3, 0) (3) {$p_3$};
\node [state] at (4, 0) (4) {$p_4$};

\node [state] at (6, 0) (q1) {$q_1$};
\node [state] at (7, 0) (q2) {$q_2$};
\node [state] at (8, 0) (q3) {$q_3$};
\node [state] at (9, 0) (q4) {$q_4$};

\node [] at (10, -1) (k4) {};

\node [state] at (5, -2) (f) {$f$};

\node [] at (5, -1) (ccf) {};

\node [] at (5, -3) (cf) {};

\path [-stealth, thick]

(q1) edge [bend right=30] node[above] {} (1)
(q2) edge [bend right=40] node[above] {} (2)
(q3) edge [bend right=50] node[above] {} (3)

(3) edge [bend right=30] node[above] {} (q1)
(3) edge [bend right=30] node[above] {} (q2)
(3) edge [bend right=30] node[above] {} (q3)

(cf) edge [line width=2pt] node[above] {} (f)
(f) edge [line width=2pt] node[above] {} (ccf)

(0) edge [line width=2pt] node[above] {} (c0)

(q4) edge [dashed, bend left=40] (k4)
;

\end{tikzpicture}
\end{center}

We have now fully defined $\Aa$. Let us first verify that there exists a mortal word for $\Aa$. Initially, we apply $s$ until only states in $Q \cup \{f\}$ are active.
Then we keep repeating the following process. Apply letter $c$, which in particular makes~$p_0$ active. Then keep applying $s$ until $q_k$ becomes active. If $s$ was applied~$h$ times, then $p_h$ is active. When we apply $d$, we get that the set of active states is a subset of $p_{h + 1}, \ldots, p_k, q_1, \ldots, q_h$, together with $f$. Moreover, by construction, the encoded value is decreased by one (ignoring the marker bit at the end). Repeat the described process until the encoded value is zero. This means that only $f$ and the state corresponding to the marker bit at the end are active, and they can be made inactive by shifting them to the right and applying $c$.

The figure below illustrates a few first steps of such counting. The number that we are tracking, including the last marker bit if it is present, is highlighted.

\begin{center}
\begin{tikzpicture} [node distance = 2cm]
\tikzset{every state/.style={inner sep=1pt,minimum size=1.5em}}

\node [state] at (0, -0.5) (0) {$p_0$};
\node [state] at (1, -0.5) (1) {$p_1$};
\node [state] at (2, -0.5) (2) {$p_2$};
\node [state] at (3, -0.5) (3) {$p_3$};
\node [state] at (4, -0.5) (4) {$p_4$};

\node [state] at (6, -0.5) (q1) {$q_1$};
\node [state] at (7, -0.5) (q2) {$q_2$};
\node [state] at (8, -0.5) (q3) {$q_3$};
\node [state] at (9, -0.5) (q4) {$q_4$};

\node [state] at (-1, -0.5) (f) {$f$};

\draw (-0.5,-0.25) -- (-0.5,-3.25);

\node [] at (0, -1) {1};
\node [] at (1, -1) {1};
\node [] at (2, -1) {1};
\node [] at (3, -1) {1};
\node [] at (4, -1) {1};
\node [] at (6, -1) {1};
\node [] at (7, -1) {1};
\node [] at (8, -1) {1};
\node [] at (9, -1)  (l1) {1};
\node [] at (-1, -1) {1};

\node [] at (0, -1.5) {0};
\node [] at (1, -1.5) {0};
\node [] at (2, -1.5) {0};
\node [] at (3, -1.5) {0};
\node [] at (4, -1.5) {0};
\node [] at (6, -1.5) {1};
\node [] at (7, -1.5) {1};
\node [] at (8, -1.5) {1};
\node [] at (9, -1.5) (l2) {1};
\node [] at (-1, -1.5) {1};

\node [] at (0, -2) {1};
\node [] at (1, -2) {0};
\node [] at (2, -2) {0};
\node [] at (3, -2) {0};
\node [] at (4, -2) {0};
\node [] at (6, -2) {1};
\node [] at (7, -2) {1};
\node [] at (8, -2) {1};
\node [] at (9, -2) (l3) {0};
\node [] at (-1, -2) {0};

\node [] at (0, -2.5) {0};
\node [] at (1, -2.5) {1};
\node [] at (2, -2.5) {0};
\node [] at (3, -2.5) {0};
\node [] at (4, -2.5) {0};
\node [] at (6, -2.5) {0};
\node [] at (7, -2.5) {1};
\node [] at (8, -2.5) {1};
\node [] at (9, -2.5) (l4) {1};
\node [] at (-1, -2.5)  {0};

\node [] at (0, -3) {0};
\node [] at (1, -3) {0};
\node [] at (2, -3) {1};
\node [] at (3, -3) {1};
\node [] at (4, -3) {0};
\node [] at (6, -3) {1};
\node [] at (7, -3) {0};
\node [] at (8, -3) {0};
\node [] at (9, -3) (l5) {0};
\node [] at (-1, -3)  {1};

\draw [draw=black] (5.8,-1.7) rectangle (9.2,-1.3);

\draw [draw=black] (5.8,-2.2) rectangle (8.2,-1.8);

\draw [draw=black] (6.8,-2.7) rectangle (9.2,-2.3);

\draw [draw=black] (1.8,-3.2) rectangle (6.2,-2.8);

\path [-stealth, thick]
(l1) edge [out=-10, in=10] node[right] {$s^5$} (l2)
(l2) edge [out=-10, in=10] node[right] {$c$} (l3)
(l3) edge [out=-10, in=10] node[right] {$s$} (l4)
(l4) edge [out=-10, in=10] node[right] {$d$} (l5)

;

\end{tikzpicture}
\end{center}

Now we need to show a lower bound on the length of shortest mortal words for $\Aa$. Intuitively, any divergence from the mortal word described above results either in a full reset of the counting (via reactivating all the states) or in a number larger than the current number (if the representation in the right half contains zeroes at the end just before $d$ is applied). Most of our argument is already outlined in the explanations of the actions of the letters.
Formally, let~$w$ be a mortal word for~$\Aa$. We can assume that the set of active states is $Q \cup \{f\}$, since this assumption will not make the mortal word longer. We look at the value of $\bin(w')$ for all prefixes $w'$ ending with $cs$. We claim that for two such subsequent prefixes $w'$ and $w''$, where $w'$ is shorter than $w''$, we have that $\bin(w'') \ge \bin(w') - 1$. We can assume that the whole set $Q \cup \{f\}$ does not get active again. Then after applying $w'$ we can only apply $s$ or $d$. If after some number of applications of $s$ state $q_k$ is not active and $d$ is applied, then by construction the encoded number is increased. Now $d$ cannot be applied again, so the increased number has to be shifted to the right half. If~$q_k$ is active, by the same reasoning we get $\bin(w'') = \bin(w') - 1$. We start counting from~$2^{k - 1}$ ($cs$ removes the marker bit at the end), hence the length of $w$ is at least $2^{k - 1}$. \end{proof}

\subsection{Binary case lower bound}

We now implement an idea similar to the idea for a ternary alphabet, but trading a letter for~$k$ more states. Again, we remark that standard techniques for reducing the size of the alphabet increase the number of states too much, so we have to introduce new techniques to directly simulate a counter using only two letters. We prove the following.

\begin{theorem}\label{thm-binary}
    For every $n \ge 5$ such that $n = 3k + 2$ for some integer $k$, there exists an NFA with $n$ states and $2$ letters such that the length of its shortest mortal word is at least $2^{(n - 2)/3}$.
\end{theorem}

\begin{proof}
We construct an NFA $\Aa = (P \cup R \cup Q \cup \{f\}, \{s, d\}, \Delta)$ with
\[P = \{p_0, p_1, \ldots, p_{k - 1}\}, R = \{r_0, r_1, \ldots, r_k\}, Q = \{q_1, \ldots, q_k\}.\]
Hence, we have $n = 3k + 2$. We appropriately refer to $P$ as the left part, $R$ as the middle part and $Q$ as the right part. The difference from the construction in the proof of \Cref{thm-ternary} is that, intuitively, we get rid of letter $c$. To do so, we split the role of the states in $P$ from the ternary case into two parts $P$ and $R$. The new left part $P$ now makes sure that the representation of the tracked number is shifted to the right part $Q$ and that $d$ is not applied too much, and the action of $d$ on the middle part $R$ performs the actual decrement. This time, no marker bit at the end of the representation is needed, but one of the states from $P \cup R$ will be active at all times, except for the very beginning of the counting process.

First we define the action of letter $s$, which has no conceptual differences from the action of $s$ in the ternary case. Formally,
\[
    p_i \cdot s = 
        \begin{cases}
         \{p_{i + 1}\} & \mbox{if } 0 \le i < k - 1; \\
        \{r_0\} & \mbox{if } i = k - 1; \end{cases}\qquad
    r_i \cdot s = 
        \begin{cases}
         \{r_{i + 1}\} & \mbox{if } 0 \le i < k; \\
        \{q_1\} & \mbox{if } i = k; \end{cases}
\]
\[
    q_i \cdot s = 
        \begin{cases}
         \{q_{i + 1}\} & \mbox{if } 1 \le i < k; \\
        Q \cup \{f\} & \mbox{if } i = k; \end{cases}\qquad
        f \cdot s = \{f\}.
\]

The action of letter $s$ for $k = 3$ is illustrated below:

\begin{center}
\begin{tikzpicture} [node distance = 2cm]
\tikzset{every state/.style={inner sep=1pt,minimum size=1.5em}}

\node[draw] at (0.7,-1) {action of $s$};

\node [state] at (0, 0) (0) {$p_0$};
\node [state] at (1, 0) (1) {$p_1$};
\node [state] at (2, 0) (2) {$p_2$};

\node [state] at (3.5, 0) (3) {$r_0$};
\node [state] at (4.5, 0) (q1) {$r_1$};
\node [state] at (5.5, 0) (q2) {$r_2$};
\node [state] at (6.5, 0) (q3) {$r_3$};

\node [state] at (8, 0) (r1) {$q_1$};
\node [state] at (9, 0) (r2) {$q_2$};
\node [state] at (10, 0) (r3) {$q_3$};

\node [state] at (5, -1) (f) {$f$};

\path [-stealth, thick]

(0) edge [] node[above] {} (1)
(1) edge [] node[above] {} (2)
(2) edge [] node[above] {} (3)
(3) edge [] node[above] {} (q1)
(q1) edge [] node[above] {} (q2)
(q2) edge [] node[above] {} (q3)
(q3) edge [] node[above] {} (r1)
(r1) edge [] node[above] {} (r2)
(r2) edge [] node[above] {} (r3)

(r3) edge [loop above] node[above] {} (r3)
(r3) edge [bend right=30] node[above] {} (r2)
(r3) edge [bend right=30] node[above] {} (r1)

(r3) edge [bend left=15] node[above] {} (f)

(f) edge [loop below] node[above] {} (f)
;

\end{tikzpicture}
\end{center}

Now we define the action of letter $d$. As mentioned above, we split the roles of states from the set $P$ from the ternary case between the sets $P$ and $R$. State $f$ now also plays a different role: if the counting is done correctly, $f$ is only active in the very beginning of the counting process. After $d$ is applied at least once, $p_0$ becomes active. From this moment onwards, there will always be at least one active state from the set $P \cup R$. The leftmost active state will remember by how much the representation of the tracked number is shifted. The action of $d$ forces this representation to be shifted to the right part, otherwise all the states are reactivated. After this shift is performed, there is exactly one active state in $R$, and this state remembers how many zeroes there are at the end of the tracked number. Then an application of $d$ decrements the value of the tracked number by one, just like in the ternary case (or, also like in the ternary case, makes the number larger if it was not shifted enough), and one state in $P$ gets reactivated to preserve the information about the current shift. Formally,
\[
    \begin{split}
    p_i \cdot d = 
         P \cup R \cup Q \cup \{f\} \\ \mbox{ for } 0 \le i \le k - 1;
    \end{split}\qquad
    r_i \cdot d = 
        \begin{cases}
         \{p_i\} \cup \{q_1, \ldots, q_i\} & \mbox{if } 0 \le i < k; \\
        \emptyset & \mbox{if } i = k; \end{cases}
\]
\[
    q_i \cdot d = 
        \begin{cases}
         \{r_i\} & \mbox{if } 1 \le i < k; \\
        \emptyset & \mbox{if } i = k; \end{cases}\qquad
        f \cdot d = \{p_0\}.
\]

The action of $d$ is illustrated below, with transitions from $r_0, r_1, r_3$ omitted:

\begin{center}
\begin{tikzpicture} [node distance = 2cm]
\tikzset{every state/.style={inner sep=1pt,minimum size=1.5em}}

\node[draw] at (0.7,-1) {action of $d$};

\node [] at (0, 1) (c0) {};
\node [] at (1, 1) (c1) {};
\node [] at (2, 1) (c2) {};
\node [] at (3, 1) (c3) {};

\node [state] at (0, 0) (0) {$p_0$};
\node [state] at (1, 0) (1) {$p_1$};
\node [state] at (2, 0) (2) {$p_2$};

\node [state] at (3.5, 0) (3) {$r_0$};
\node [state] at (4.5, 0) (q1) {$r_1$};
\node [state] at (5.5, 0) (q2) {$r_2$};
\node [state] at (6.5, 0) (q3) {$r_3$};

\node [state] at (8, 0) (r1) {$q_1$};
\node [state] at (9, 0) (r2) {$q_2$};
\node [state] at (10, 0) (r3) {$q_3$};

\node [] at (11, -1) (cr3) {};

\node [state] at (5, -1) (f) {$f$};

\path [-stealth, thick]

(r2) edge [bend right=40] node[above] {} (q2)
(r1) edge [bend right=30] node[above] {} (q1)

(q2) edge [bend right=30] node[above] {} (r2)
(q2) edge [bend right=30] node[above] {} (r1)
(q2) edge [bend left=30] node[above] {} (2)

(f) edge [bend left=15] node[above] {} (0)

(0) edge [line width=2pt] node[above] {} (c0)
(1) edge [line width=2pt] node[above] {} (c1)
(2) edge [line width=2pt] node[above] {} (c2)

(r3) edge [dashed, bend left=40] (cr3)

;

\end{tikzpicture}
\end{center}

The proof that thus constructed NFA $\Aa$ admits a mortal word is very similar to that in \Cref{thm-ternary}: if all shifts are performed correctly (that is, if $Q \cup \{f\}$ never gets reactivated) and fully, then the tracked number decreases by one with each application of $d$, until it becomes zero and we can use $d$ to kill the remaining active state from $P$. A few first steps are illustrated below, with the tracked number highlighted:

\begin{center}
\begin{tikzpicture} [node distance = 2cm]
\tikzset{every state/.style={inner sep=1pt,minimum size=1.5em}}

\node [state] at (0, 0) (0) {$p_0$};
\node [state] at (1, 0) (1) {$p_1$};
\node [state] at (2, 0) (2) {$p_2$};

\node [state] at (3.5, 0) (3) {$r_0$};
\node [state] at (4.5, 0) (q1) {$r_1$};
\node [state] at (5.5, 0) (q2) {$r_2$};
\node [state] at (6.5, 0) (q3) {$r_3$};

\node [state] at (8, 0) (r1) {$q_1$};
\node [state] at (9, 0) (r2) {$q_2$};
\node [state] at (10, 0) (r3) {$q_3$};

\node [state] at (-1, 0) (f) {$f$};

\draw (-0.5,0.25) -- (-0.5,-2.75);

\node [] at (0, -0.5) {1};
\node [] at (1, -0.5) {1};
\node [] at (2, -0.5) {1};

\node [] at (3.5, -0.5) {1};
\node [] at (4.5, -0.5) {1};
\node [] at (5.5, -0.5) {1};
\node [] at (6.5, -0.5)  {1};

\node [] at (8, -0.5) {1};
\node [] at (9, -0.5) {1};
\node [] at (10, -0.5) (l1) {1};

\node [] at (-1, -0.5) {1};

\node [] at (0, -1) {0};
\node [] at (1, -1) {0};
\node [] at (2, -1) {0};

\node [] at (3.5, -1) {0};
\node [] at (4.5, -1) {0};
\node [] at (5.5, -1) {0};
\node [] at (6.5, -1)  {0};

\node [] at (8, -1) {1};
\node [] at (9, -1) {1};
\node [] at (10, -1) (l2) {1};

\node [] at (-1, -1) {1};

\node [] at (0, -1.5) {1};
\node [] at (1, -1.5) {0};
\node [] at (2, -1.5) {0};

\node [] at (3.5, -1.5) {0};
\node [] at (4.5, -1.5) {1};
\node [] at (5.5, -1.5) {1};
\node [] at (6.5, -1.5)  {0};

\node [] at (8, -1.5) {0};
\node [] at (9, -1.5) {0};
\node [] at (10, -1.5) (l3) {0};

\node [] at (-1, -1.5) {0};

\node [] at (0, -2) {0};
\node [] at (1, -2) {0};
\node [] at (2, -2) {0};

\node [] at (3.5, -2) {0};
\node [] at (4.5, -2) {1};
\node [] at (5.5, -2) {0};
\node [] at (6.5, -2)  {0};

\node [] at (8, -2) {0};
\node [] at (9, -2) {1};
\node [] at (10, -2) (l4) {1};

\node [] at (-1, -2) {0};

\node [] at (0, -2.5) {0};
\node [] at (1, -2.5) {1};
\node [] at (2, -2.5) {0};

\node [] at (3.5, -2.5) {0};
\node [] at (4.5, -2.5) {0};
\node [] at (5.5, -2.5) {1};
\node [] at (6.5, -2.5)  {0};

\node [] at (8, -2.5) {1};
\node [] at (9, -2.5) {0};
\node [] at (10, -2.5) (l5) {0};

\node [] at (-1, -2.5) {0};

\draw [draw=black] (7.8,-0.8) rectangle (10.2,-1.2);

\draw [draw=black] (4.3,-1.3) rectangle (6.7,-1.7);

\draw [draw=black] (8.8,-1.8) rectangle (10.2,-2.2);

\draw [draw=black] (5.3,-2.3) rectangle (8.2,-2.7);

\path [-stealth, thick]
(l1) edge [out=-10, in=10] node[right] {$s^7$} (l2)

(l2) edge [out=-10, in=10] node[right] {$d$} (l3)

(l3) edge [out=-10, in=10] node[right] {$s^4$} (l4)

(l4) edge [out=-10, in=10] node[right] {$d$} (l5)

;

\end{tikzpicture}
\end{center}

To prove a lower bound on the length of a shortest mortal word, we formally look at $\bin(w')$ for prefixes $w'$ of an applied word $w$ such that state $p_0$ is active after the application of $w'$. As in the ternary case, we argue that for two such subsequent prefixes $w'$ and $w''$ we have that $\bin(w'') \ge \bin(w') - 1$ by construction. Thus, the length of a shortest mortal word is at least~$2^k$.  
\end{proof}

%% file: sec-dfas.tex
The presented constructions, as well as all mentioned constructions from the literature, have purely combinatorial nature. In this brief section we discuss a possibility of a more algebraic view on automata mortality. We focus on a simpler case of strongly connected DFAs, since even for this case very little is known. A~DFA is called \emph{strongly connected} if for every pair $p, q$ of its states there is a word mapping $p$ to $q$.

In \cite{Ananichev2013}, another reachability property of strongly connected complete DFAs, namely the reset threshold, was linked to exponents of primitive matrices. A word is called \emph{synchronizing} for a complete DFA $\Aa$ if it maps all states of $\Aa$ to the same state. The length of a shortest synchronizing word of $\Aa$ is then called its \emph{reset threshold}. Let $M$ be the sum of the matrices of all letters of~$\Aa$. If~$\Aa$ is synchronizing, $M$ must be primitive. A matrix is \emph{primitive} if some its power has only strictly positive entries. The smallest such power is called the \emph{exponent} of a matrix. In \cite{Ananichev2013} it is proved that the exponent of $M$ is at most the reset threshold of $\Aa$ plus $n - 1$, where~$n$ is the number of states of~$\Aa$. Thus, one obstacle for a complete DFA to synchronize fast is a large exponent of the corresponding matrix. A more advanced property of similar kind is proposed in~\cite{Gusev2013}.

The DFA mortality problem is sometimes framed as synchronization of a particular class of complete DFAs (namely, of complete DFAs with a single sink state), but we find that DFA mortality significantly differs from DFA synchronization and deserves independent treatment. One reason for that is the dependence on the size of the alphabet: adding more letters to a complete DFA usually helps to synchronize it faster, but, in contrast, it allows a DFA to ``stall'' for longer when applying a mortal word. Thus, the known series of complete DFAs with the largest reset threshold have two letters \cite{Ananichev2013}, but the series with the largest known mortality threshold have the number of letters equal to the number of states \cite{Rystsov1997}. Perhaps more importantly, the results of \cite{Ananichev2013} do not apply to DFA mortality, since after adding a sink state a DFA stops being strongly connected.

Mortality and synchronization can be viewed as very combinatorial phenomena: having a set of matrices, we ask for just one product of such matrices with a certain property, and hence have to make a lot of independent choices when constructing this product. Primitivity of a single matrix has much more algebraic flavour to it, since it concerns iterated multiplication of a single matrix.

We thus pose the following semi-formal question: \emph{is there an algebraic property of a DFA telling us something about its mortality threshold?}

A simpler question is to ask for an algebraic property guaranteeing that a mortal word exists. For DFAs this is trivial, but for unambiguous NFAs, a proper superclass of DFAs, such property is precisely that the spectral radius of the average of the generating matrices is strictly smaller than one \cite{Kiefer2021}. For general NFAs one has to keep in mind that deciding the existence of a mortal word becomes \PSPACE-complete, and hence one has to look for a ``non-constructive'' property.


The discovery of the connection between reset threshold and the exponent of a matrix in \cite{Ananichev2013} was done by observing in exhaustive search experiments the similarity between the attainable large values of reset thresholds on the one hand and exponents of primitive matrices on the other hands.
The present paper can be seen as a step towards a similar potential development for mortality. For small~$n$, it is not hard to write down the sets of $n \times n$ matrices with large mortality thresholds provided in the previous section, in case one wants to check some conjectures connected to our question. To further assist with this goal, we now briefly survey what is known about DFA mortality, and provide a new simple family of binary DFAs with large mortality thresholds.

Denote by $\mtd(n, m)$ the maximum mortality threshold of a DFA with $n$ states and $m$ letters admitting a mortal word. In \cite{Rystsov1997} it is proved that for all~$m$ we have  $\mtd(n, m) \le \frac{n(n + 1)}{2}$, and that this upper bound is tight already for $m = n$. The binary case is more intriguing. To the best of our knowledge, the first series of $n$-state binary DFAs with quadratic mortality threshold is provided in Proposition 3.4 of \cite{Beal2003}, showing that $\mtd(n, 2) \ge \frac{(n - 1)^2}{4}$. Later on, such series with slightly larger mortality threshold were provided in \cite{Martugin2008} and \cite{Ananichev2019}, showing, respectively,
\[\mtd(n, 2) \ge \frac{n^2 + 8n}{4} + \OO(1) \quad \mbox{and} \quad \mtd(n, 2) \ge \frac{n^2 + 10n}{4} + \OO(1).\]

A few more series with simpler structure but with slightly smaller mortality threshold are provided in \cite{Ryzhikov2019WORDS}. It is not known if $\mtd(n, 2) = \frac{n^2}{4} + \OO(n)$.

Below we provide a very simple series of binary DFAs with mortality threshold $\frac{n^2 + 8n}{4} + \OO(1)$ which does not seem to appear in the literature. It can be seen as a simplification of the construction from \cite{Martugin2008} with approximately the same mortality threshold.
Let $\Aa = (Q, \{a, b\}, \delta)$ be such that $Q = \{q_1, \ldots, q_k, p_1, \ldots, p_k\}$. Define 
\[
    q_i \cdot a = 
        \begin{cases}
        q_{i + 1} & \mbox{if } 1 \le i < k;\\
        q_1 & \mbox{if } i = k;\end{cases}\quad
    p_i \cdot a = 
        \begin{cases}
         p_{i + 1} & \mbox{if } 1 \le i < k; \\
        \emptyset & \mbox{if } i = k; \end{cases}
\]
\[
    q_i \cdot b = 
        \begin{cases}
        p_1 & \mbox{if } i = 1;\\
        q_{k} & \mbox{if } i = 2;\\
        q_{i - 1} & \mbox{if } 2 < i \le k;
        \end{cases}\quad
    p_i \cdot b = 
         q_1 \mbox{ for } 1 \le i \le k.
\]

The fact that $\Aa$ has mortality threshold $\frac{n^2 + 8n}{4} + \OO(1)$ follows, for example, from Lemma 1 in \cite{Ananichev2019}, since this is another example of adding a tail to an almost permutation automaton, formally defined in \cite{Ananichev2019}. While most of the known series with large mortality thresholds have this structure, not all of them do: for example, DFAs from the series in \cite{Beal2003} do not have a tail, and have a cycle going through all states instead.

An illustration for the case $k = 4$ is provided below, with the action of $a$ depicted by solid arrows and the action of $b$ by dashed arrows.

\input{fig-series2}

The unique shortest mortal word for this DFA is $a^k b a^k a b a^k (a a b a^k)^{k - 2}$. The intuition behind this word is that if everything is done ``optimally'', after killing (that is, applying a word that is undefined for a state) two states among $q_1, \ldots, q_k$, the structure of the cycles on them maintains that $q_1$ and $q_k$ are inactive, and hence they have to be traversed in the process of killing every remaining active state from $q_1, \ldots, q_k$.

%% file: fig-series2.tex
\begin{center}
\begin{tikzpicture} [node distance = 2cm]
\tikzset{every state/.style={inner sep=1pt,minimum size=1.5em}}

\node [state] at (1.5, 0) (1) {$q_1$};
\node [state] at (1.5, 1.5) (2) {$q_2$};
\node [state] at (0, 1.5) (3) {$q_3$};
\node [state] at (0, 0) (4) {$q_4$};

\node [state] at (4, -1) (5) {$p_1$};
\node [state] at (6, -1) (6) {$p_2$};
\node [state] at (8, -1) (7) {$p_3$};
\node [state] at (10, -1) (8) {$p_4$};

\node [] at (11, -2) (f) {};

\path [-stealth, thick]

(1) edge [bend right=20] node[above] {} (2)
(2) edge [bend right=20] node[above] {} (3)
(3) edge [bend right=20] node[above] {} (4)
(4) edge [bend right=20] node[above] {} (1)

(1) edge [dashed, bend right=20] node[above] {} (5)
(5) edge [] node[above] {} (6)
(6) edge [] node[above] {} (7)
(7) edge [] node[above] {} (8)
(8) edge [bend left=40] node[above] {} (f)

(5) edge [dashed, bend right=20] node[above] {} (1)
(6) edge [dashed, bend right=25] node[above] {} (1)
(7) edge [dashed, bend right=30] node[above] {} (1)
(8) edge [dashed, bend right=30] node[above] {} (1)

(2) edge [dashed, bend left=20] (4)
(3) edge [dashed, bend right=20] (2)
(4) edge [dashed, bend right=20] (3)

;



\end{tikzpicture}
\end{center}

%% file: sec-conclusions.tex
One of the most interesting questions that remains open is the precise behaviour of $\mt(n, m)$ and $\mtd(n, m)$ for $m = 2$ and $m = 3$. As already discussed in \Cref{sec-dfas}, lower bounds on $\mtd(n, 2)$ were investigated quite actively \cite{Ananichev2019,Beal2003,Martugin2008,Ryzhikov2019WORDS}, and all these bounds are still of the form $\frac{n^2}{4} + \OO(n)$. To the best of our knowledge, there is no plausible conjecture on the value of $\mtd(n, 2)$ in the literature, and there are no known upper bounds better than $\frac{n(n + 1)}{2}$ \cite{Rystsov1997}, a bound that holds for any size of the alphabet. It is in fact easy to see that this bound cannot be tight for $\mtd(n, 2)$, but we were not able to obtain a significant improvement for it. However, we conjecture that it can be improved to at least the bound $\mtd(n, 2) \le \frac{7n^2}{16} + \OO(n)$. As for bounds on $\mtd(n, 3)$, we are not aware of any results. Since the dependence of $\mtd(n, m)$ on $m$ seems to be far from trivial, this is a natural first step to approach its investigation. Alternatively, one can ask for bounds on $\mtd(n, \frac{n}{2})$.

The same questions can be asked about $\mt(n, m)$, and the situation with existing results is similar: we already surveyed known lower bounds in \Cref{sec-existing} and mentioned that $\mt(n, m) \le 2^{n} - 1$. To the best of our knowledge, no better upper bounds for small $m$ are known, and the same comments that we made about the precise dependence of $\mtd(n, m)$ on $m$ also apply to $\mt(n, m)$.

Mortality of unambiguous NFAs is even further from being understood. As mentioned in the introduction, we know that the mortality threshold for them is between $\frac{n(n+1)}{2}$ \cite{Rystsov1997} and $n^5$ \cite{Kiefer2021}. The lower bound is provided by a series of DFAs, and all known families of unambiguous NFAs with large mortality threshold are DFAs or their co-deterministic counterparts (obtained by reversing the direction of each transition). Since unambiguous NFAs constitute a much more general class than DFAs, this is a curious situation that deserves further investigation.

Finally, it would be interesting to see any subclasses with good (that is, polynomial) behaviour of $\mt(n, m)$ besides the class of unambiguous NFAs \cite{Kiefer2021}.